\documentclass[12pt]{article}

\usepackage{amsmath,amsfonts,amsthm,amscd,amssymb,slashed}

\usepackage[pdftex]{graphicx}
\usepackage{algpseudocode}
\usepackage{tikz}
\usepackage{enumerate}
\usepackage[hidelinks]{hyperref}
\usepackage{mathtools}
\usetikzlibrary{shapes,arrows}
\usetikzlibrary{matrix,arrows}
\usepackage{epstopdf}
\usepackage{slashed}
\usepackage{xcolor,import}
\usepackage{transparent}
\usepackage[all]{xy}

\usepackage{geometry}
\makeatletter

\@addtoreset{equation}{section}
\def\section{\@startsection {section}{1}{\z@}{-2.5ex plus -1ex minus
 -.2ex}{1.3ex plus .2ex}{\large\bf}}
\def\subsection{\@startsection{subsection}{2}{\z@}{-2.25ex plus%
 -1ex minus -.2ex}{0.5ex plus .2ex}{\bf}}

\newcommand{\R}{\mathbb{R}}
\newcommand{\N}{\mathbb{N}}
\newcommand{\C}{\mathbb{C}}
\newcommand{\Z}{\mathbb{Z}}

\newcommand{\Il}{\mathcal{I}_{\ell}}
\newcommand{\AdS}{\text{AdS}_3}
\newcommand{\AAdS}{\widetilde{\text{AdS}}_3}
\newcommand{\bee}{\begin{equation}}
\newcommand{\eee}{\end{equation}}
\newtheorem{thm}{Theorem}[section]
\newtheorem{prop}[thm]{Proposition}
\newtheorem{cor}[thm]{Corollary} 
\newtheorem{dfn}[thm]{Definition}
\newtheorem{lem}[thm]{Lemma}
\theoremstyle{definition}

\newcommand{\lh}{\pi}

\newcommand{\lf}[1]{\sigma^{ #1}}

\numberwithin{equation}{section}

\begin{document}

\begin{flushright}
EMPG-18-10
\end{flushright}
\vskip 10pt
\baselineskip 28pt

\begin{center}
{\Large \bf Hyperbolic vortices and Dirac fields in 2+1 dimensions}

\baselineskip 18pt

\vspace{0.4 cm}

{\bf Calum Ross and Bernd J.~Schroers}\\
\vspace{0.2 cm}
Maxwell Institute for Mathematical Sciences and
Department of Mathematics,
\\Heriot-Watt University,
Edinburgh EH14 4AS, UK. \\
{\tt cdr1@hw.ac.uk} and {\tt b.j.schroers@hw.ac.uk} \\

\vspace{0.4cm}

{June 2018} 
\end{center}

\baselineskip 16pt
\parskip 4 pt

\parindent 10pt

\begin{abstract}
\noindent 
Starting from the geometrical interpretation of integrable vortices on two-dimensional hyperbolic space as conical singularities,  we explain how this picture can be expressed in the language of Cartan connections, and how it  can be lifted to the double cover of three-dimensional  Anti-de Sitter space viewed as a trivial circle bundle over hyperbolic space.  
We show that vortex configurations  on the double cover of AdS  space  give  rise to solutions of the Dirac equation  minimally coupled to the magnetic field of the vortex. After stereographic projection to (2+1)-dimensional Minkowski space we  obtain, from each lifted hyperbolic vortex, a  Dirac field and an abelian gauge field  which solve a   Lorentzian, (2+1)-dimensional version of the Seiberg-Witten equations. 
\end{abstract}

\noindent

\section{Introduction}

Vortex configurations consisting of a complex scalar and an abelian gauge field  can be given a geometrical interpretation 
 in two dimensions by viewing the modulus of the scalar field as a conformal rescaling of the  underlying two-dimensional geometry. As pointed out in \cite{Baptista}, this  is  particularly natural for vortices defined on  K\"ahler  geometries and obeying first-order Bogomol'nyi equations. In that case, the  rescaled metric has conical singularities at the vortex locations but away from those singularities its curvature can be expressed very simply in terms  of the  original K\"ahler form and curvature, and the rescaled  K\"ahler form. 
 
 When the first order vortex equations are integrable, the geometrical picture   simplifies further.  The  Popov vortex equations on the sphere \cite{Popov}, for example, can be solved in terms of a rational map \cite{Manton1}. The rescaled metric is then the pullback of the round metric on the sphere via the rational map. It  has conical singularities at the ramification points, and these  are precisely the locations of the vortices. 
 
 Recently, we showed in \cite{RS} that  one gains further insight into the geometry of   Popov vortices  by considering the  total space $S^3\simeq SU(2)$ of the circle  bundle on which they are defined. We showed that the equivariant formulation of the vortex equation on $S^3$ can be solved in terms of bundle maps  of the Hopf fibration,  and that each vortex configuration gives rise to a solution of the gauged Dirac equation on $S^3$ and, after stereographic projection, on Euclidean space.  In this way, equivariant Popov vortices on $SU(2)$  can be seen to generate manifestly square integrable and smooth solutions of the magnetic zero-mode problem posed
 and studied by Loss and Yau \cite{LY}. This picture also  provides a geometrical context for earlier results  on magnetic zero-modes by Adam, Muratori and Nash \cite{AMN1,AMN2}.

The goal of the present work is to  apply the geometrical concepts and  methods   of  \cite{RS} to  the integrable vortex equations on hyperbolic space  $H^2$, described here in terms of the Poincar\'e disk model. The role of the total space is played by the double cover   $\AAdS$ of Anti-de Sitter space, which is isomorphic to $SU(1,1)$.  The  analogue of the Hopf fibration is the (topologically trivial)  circle fibration   $\lh:SU(1,1)\to H^{2}$.  Proceeding as in \cite{RS}, we  show that hyperbolic vortices can be lifted to vortex configurations on $\AAdS$, and that they give rise to solutions of the gauged Dirac equation on $\AAdS$ and a non-linear constraint.   Using a  suitably defined stereographic projection we finally obtain, from each lifted hyperbolic vortex, a   solution of a  Lorentizan and (2+1)-dimensional version of the Seiberg-Witten equations,  consisting of 
the gauged Dirac equation  and an equation relating the magnetic field to the spin density. 

The standard geometry of the three dimensional group manifold $SU(1,1)$ plays a central role in this paper. The lifted vortex configurations can be expressed in terms of  bundle maps of $\lh:SU(1,1)\to H^{2}$,  covering holomorphic maps $H^{2}\to H^{2}$. These bundle maps can in turn be written in terms of two holomorphic functions satisfying an equivariance condition; they  fully determine the vortices, the magnetic fields and the  Dirac fields.

As in \cite{RS}, we have found it illuminating to express the geometry defined by a vortex configuration in  the language of Cartan geometry. This provides the simplest route to three dimensions, and relates the vortex equations to a flatness condition for a non-abelian connection.  The fact that $SU(1,1)$ is a trivial circle bundle over  $H^{2}$ simplifies the discussion, but the non-trivial first fundamental group of $SU(1,1)$ and  non-compactness of $H^2$ add subtleties compared to the Euclidean story, as we will explain.

The paper is organised as follows. 
We begin in Sect.~2 by summarising some results on hyperbolic vortices and their geometric interpretation. Then, following a brief interlude to establish our conventions for the pseudo-unitary group $SU(1,1)$, we present a Cartan geometry interpretation of hyperbolic vortices. In particular we construct an explicit $su(1,1)$ gauge potential  whose flatness is equivalent to the hyperbolic vortex equations.

Sect.~3 introduces the three dimensional setting, the relationship between $\AdS$ and the Lie group $SU(1,1)$ and an interpretation of both  as  circle bundles over $H^{2}$. The main result of this section is the equivalence between vortex configurations on $SU(1,1)$ and flat $SU(1,1)$ gauge potentials,  and  expressions for both of these in terms of bundle maps $SU(1,1)\to SU(1,1)$. At this stage the hyperbolic story is arguably richer than its spherical counter part and we are not immediately forced to consider vortex configurations of finite degree. However, configurations  with finite equivariant degree  give rise to finite charge vortices on $H^{2}$, 
and we therefore pay them special attention. 

In Sect.~4 we introduce stereographic and gnomonic projections from $SU(1,1)$ to three-dimensional Minkowski space $\R^{1,2}$, and use them  
to relate the  gauged Dirac equation on  $SU(1,1)$ and  on the interior $\Il\subset\R^{1,2}$ of a single-sheeted hyperboloid. 
Then, in Sect.~5, we  combine  the results of all preceding sections to construct solutions of a Lorentzian version of the Seiberg-Witten equations  on $\AAdS$ and on (2+1)-dimensional  Minkowski space.
Finally Sect.~6 contains a summary  of the paper in the form of   Fig.~\ref{summary}, and a discussion.

\section{Hyperbolic vortices and Cartan geometry}

\subsection{Hyperbolic vortices and holomorphic maps}
\label{hyphol}
First order vortex equations  on the Poincar\'e disk model of hyperbolic space have been studied extensively,  with many of the details summarised in \cite{MS}. Solutions can be obtained from $SO(3)$ invariant instantons on $S^{2}\times H^{2}$ \cite{Witten1} and expressed in terms of holomorphic mappings of the disk.   We briefly review these solutions here, but should warn the reader  
  that our  disk  has radius $1$ rather than $\sqrt{2}$  as chosen in \cite{MS}, and  that we write the equations in terms of the Riemann curvature form rather than the K\"ahler form for the disk. 
  
  We first introduce our notation for the geometry of the Poincar\'e disk model, which we write as 
\begin{equation}
H^{2}=\{z\in \C| |z|^{2}<1\}, \label{diskmodeldef}
\end{equation}
with  metric 
\begin{equation}
ds^{2}=\frac{4 dzd\bar{z}}{(1-|z|^{2})^{2}}. \label{H^2 metric}
\end{equation}
A complexified orthonormal frame field $e=e_{1}+ie_{2}$ for this metric is
\begin{equation}
e=\frac{2}{1-|z|^{2}}dz, \quad \bar{e}=\frac{2}{1-|z|^{2}}d\bar{z}. \label{complex frame}
\end{equation}
This metric is K\"ahler with K\"ahler form
\begin{equation}
\omega=\frac{i}{2}e\wedge\bar{e}=2i\frac{dz\wedge d\bar{z}}{(1-|z|^{2})^{2}}. \label{kahler form}
\end{equation}
In terms of the complexified frame the structure equations are given by
\begin{equation}
de-i\Gamma\wedge e=0 \label{structureeq}
\end{equation}
and its complex conjugate.
The structure equations determine the spin connection $1$-form $\Gamma$ to be
\begin{equation}
\Gamma=i\frac{\bar{z}dz-zd\bar{z}}{1-|z|^{2}}. \label{spinconnection}
\end{equation}
The Riemann curvature form is 
\begin{equation}
\mathcal{R}=d\Gamma,
\end{equation}
and is related to the Gauss curvature, $K$, and the K\"ahler form, equation \eqref{kahler form}, through the Gauss equation,
\begin{equation}
\mathcal{R}=K\omega, \label{Gauss}
\end{equation}
with $K=-1$ for $H^{2}$. 

A hyperbolic vortex is a pair $(\phi,a)$ where $a$ is a connection on a principal $U(1)$ bundle over $H^{2}$, which is necessarily trivial, and $\phi$ is a smooth section of the associated complex line bundle. Taking $a=a_{z}dz+a_{\bar{z}}d\bar{z}$ and $F_{a}=da$, the vortex equations are
\begin{equation}
\partial_{\bar{z}}\phi-ia_{\bar{z}}\phi=0, \quad F_{a}=(|\phi|^{2}-1)\mathcal{R}. \label{hyperbolic vortex}
\end{equation} 
The first of these requires that $\phi$ be holomorphic with respect to the connection $a$.

It is instructive to compare this equation to a  different vortex equation, called the Popov vortex equation, which was introduced in \cite{Popov} and studied in \cite{Manton1}. A Popov vortex is a pair $(\phi, a)$ of a connection $a$ on a degree $2N-2$ principal $U(1)$-bundle  over the two-sphere and a holomorphic section $\phi$ of an associated line bundle. In terms of  a stereographic coordinate $z$ and curvature two-form $\mathcal{R}_{S^{2}}$,  the  equations are
\begin{equation}
\partial_{\bar{z}}\phi-ia_{\bar{z}}\phi=0, \quad F_{a}=\left(|\phi|^{2}-1\right)\mathcal{R}_{S^{2}}. \label{Popov vortex equation}
\end{equation}
 In the form presented here, these equations look the same  as the hyperbolic equations  since  the sign difference due to the Gauss curvature is absorbed into the Riemann curvature two-form. From now on we work with  hyperbolic vortices but make frequent comparisons with Popov vortices.

Solutions to the hyperbolic vortex equations can be constructed from holomorphic functions $f:H^{2}\to H^{2}$ in the following way. Given the complex frame and spin connection, $e,\Gamma$ on $H^{2}$, we  pull them back via $f$ and define the Higgs field and gauge potential through
\begin{equation}
\phi e = f^{*}e, \quad a=f^{*}\Gamma-\Gamma,\label{vortexpair}
\end{equation}
so that we have an explicit expression for $\phi$ in terms of $f$:
\bee
\label{phif}
\phi=\frac{1-|z|^{2}}{1-|f|^{2}}f'.
\eee
It  also follows that
\begin{equation}
f^{*}(e\wedge \bar{e})=|\phi|^{2}e\wedge\bar{e},
\end{equation}
so that the second vortex equation is a direct consequence:
\begin{equation}
F_{a}=d(f^{*}\Gamma-\Gamma)=f^{*}\mathcal{R}-\mathcal{R}=(|\phi|^{2}-1)\mathcal{R}.
\end{equation}
To see that $\phi$ is indeed covariantly holomorphic with respect to $a$ as the first vortex equation requires, consider the pullback of the structure equation, \eqref{complex structure eq}:
\begin{align}
0	&=df^{*}e-if^{*}\Gamma\wedge f^{*}e \nonumber \\
	&=\left(de-i\Gamma\wedge e\right)\phi+\left(d\phi-i\left(f^{*}\Gamma-\Gamma\right)\phi\right)\wedge e \nonumber \\
	&=\left(d\phi-ia\phi\right)\wedge e.
\end{align}
The final line is  the  required holomorphicity condition, and equivalent to the first vortex equation.

The pulled back frame $f^{*}e$ degenerates  at the zeros of $\phi$. This corresponds to the vortex positions becoming conical singularities of the rescaled metric,  with an angular excess related to the charge of the vortex, see  \cite{Manton2}.  A consequence of the frame being degenerate is that its spin connection, $\tilde{\Gamma}$, has singularities. However, $f^{*}\Gamma$ is the pullback of a smooth spin connection with a smooth map,  so is in particular  non-singular. The difference between the pulled back spin connection and the spin connection of the degenerate frame is due to singularities at the zeros of $\phi$. This results in $\tilde{\mathcal{R}}$ being equal to $f^{*}\mathcal{R}$ up to the addition of delta function singularities at the zeros $z_{j}$ of $\phi$:
\begin{equation}
\label{riemanndelta}
\tilde{\mathcal{R}}=f^{*}\mathcal{R}-2\pi\sum_{j}\delta_{z_{j}}.
\end{equation}

One can view Riemann surfaces of genus $g>1$  as the quotient of $H^{2}$ by  the action of a  Fuchsian group $\Gamma< SU(1,1)$. Vortices on  such Riemann surfaces can therefore  be constructed  from vortices on $H^{2}$ that are invariant under the action of the desired Fuchsian group. In practice, this is not easy. A vortex solution on the Bolza surface (genus $2$) is presented in \cite{MM}. While these vortices have an infinite number of zeros of the Higgs field on $H^{2}$ they have a finite number of zeros within the principal domain of the Fuchsian group.

In the construction of   vortices from  holomorphic maps between compact Riemann surface via  \eqref{vortexpair}, the Gauss-Bonnet theorem imposes a constraint on the genus of the surfaces and the vortex number, see  \cite{Baptista}. The negative contribution from the  singularities in the curvature  $\tilde{\mathcal{R}}$, \eqref{riemanndelta},  at the zeros of $\phi$ plays a key role here, and provides a no-go theorem in some cases.

We now specialise to the case of solutions of the vortex equations on   $H^{2}$ which satisfy the boundary condition $|\phi|\to 1$ as $|z|\to 1$. As explained in \cite{MS, JT}, this boundary condition is required for the vortex to have finite energy. It also means that $\phi$  has a finite  vortex  number associated to it,  which is the number of zeros counted with multiplicity. This is analogous to the degree of the line bundle for the case of Popov vortices mentioned above. 

As was first observed in \cite{Witten1}, solutions of the hyperbolic vortex equations on $H^{2}$ which satisfy the boundary condition are obtained from bounded holomorphic functions $f:H^{2}\to H^{2}$ which can be expressed as a finite Blaschke product when working in the Poincar\'e disk model. 
For a $(N-1)$-vortex solution,  the Blaschke product  can be written as the ratio
\begin{equation}
f=\frac{f_{2}}{f_{1}} \label{fblaschke}
\end{equation}
of the two holomorphic functions 
\begin{equation}
 f_{1}(z)=\prod_{k=1}^{N}\left(1-\bar{c}_k z\right), \quad f_{2}(z)=\prod_{k=1}^{N}\left(z-c_k\right),
\end{equation}
where $c_k\in H^{2}$, $k=1,\ldots,N$.
As the zeros of $f_{2}$ are in the disk of radius $1$ the zeros of $f_{1}$, at $1/\bar{c}_k$, are not. Thus $f$ has zeros but no poles within the disk. Note that $|f(z)|=1$ when $|z|=1$, i.e., on the boundary of the disk and that, by the maximum principle,  $|f(z)|<1$ when $|z|<1$, so that $f$ really is a holomorphic mapping of the disk model.

This way of writing the holomorphic function  will prove useful later when we introduce and discuss vortex configurations on $SU(1,1)$,  as will the observation about the lack of poles.
The  pullback of the holomorphic frame field  has the explicit form 
\begin{equation}
 f^*e= \phi e=2i\frac{f_{2}'f_{1}-f_{1}'f_{2}}{|f_{1}|^{2}-|f_{2}|^{2}}\frac{\bar{f}_{1}}{f_{1}}dz.
\end{equation}
This is a manifestly  smooth one-form which vanishes at the zeros of $\phi$,  thus  illustrating our earlier remarks about the pullback frame.

\subsection{Interlude on \texorpdfstring{$SU(1,1)$}{SU(1,1)}}
\label{su11conventions}
Before we interpret hyperbolic vortices in terms of Cartan geometry we need to make clear our conventions for the pseudo-unitary group $SU(1,1)$. It is defined as the subgroup of $SL(2,\C)$ whose elements $h$ satisfy
\begin{equation}
h\tau_{3}h^{\dagger}=\tau_{3},
\end{equation}
where $\tau_{3}$ is the third Pauli matrix.
Its Lie algebra $su(1,1)$ is defined as the  set of complex traceless matrices satisfying
\begin{equation}
g^{\dagger}=-\tau_{3}g\tau_{3}.
\end{equation}
This forces the diagonal elements to be purely imaginary and the off-diagonal elements to be mutually complex-conjugate. 
We work with the generators
\begin{equation}
t_{0}=\frac{i}{2}\tau_{3}, \quad t_{1}=-\frac{1}{2}\tau_{2}, \quad t_{2}=\frac{1}{2}\tau_{1},
\end{equation}
where the $\tau_{i}$ are the Pauli matrices. They obey the commutation relation
\begin{equation}
[t_{i},t_{j}]=\varepsilon_{ij}^{\;\;\;k}t_{k}.
\end{equation}
We also frequently use
\begin{equation}
t_{+}=t_{1}+it_{2}, \quad t_{-}=t_{1}-it_{2},
\end{equation}
which have commutators
\begin{equation}
[t_{+},t_{-}]=-2it_{0}, \quad [t_{0},t_{\pm}]=\pm it_{\pm},
\end{equation}
showing that  $t_{0}$ acts as  a complex structure on its complement in $su(1,1)$, with $t_{+}(t_{-})$  as (anti)-holomorphic directions. 

The Killing form on $su(1,1)$ is
\begin{equation}
\kappa_{ij}=\kappa(t_{i},t_{j})=\frac{1}{2}\varepsilon_{ikm}\varepsilon_{j}^{\;\;km}=\eta_{ij},
\end{equation} 
where $\eta$ is the `mostly minus' Minkowski metric:
\begin{equation}
\label{etadef}
\eta=\text{diag}(1,-1,-1).
\end{equation}

We parametrise an $SU(1,1)$ matrix $h$ using  complex coordinates $(z_1,z_2)\in \C^{1,1}$ as
\begin{equation}
h=\begin{pmatrix}
z_1& \bar{z}_{2}\\
z_2&\bar{z}_1
\end{pmatrix},
 \qquad |z_1|^{2}-|z_2|^{2}=1,
\label{complex coordinate matrix}
\end{equation}
so that we can view $SU(1,1)$ as a submanifold of $\C^{1,1}$ :
\begin{equation}
SU(1,1)=\{(z_1,z_2)\in \C^{1,1}||z_1|^{2}-|z_2|^{2}=1\}. \label{su(1,1) complex def}
\end{equation}
It is the double cover of $\AdS$, the real submanifold of $\C^{1,1}$ defined by
\begin{equation}
\AdS=\{(z_1,z_2)\in \C^{1,1}| |z_1|^{2}-|z_2|^{2}=\ell^{2}\}/\Z_{2} , \label{Ads cdef}
\end{equation}
with $\ell$ called the AdS length. The $\Z_{2}$ quotient identifies $(-z_1,-z_2)$ and $(z_1,z_2)$.

Left invariant one-forms on $SU(1,1)$ are defined via
\begin{equation}
h^{-1}dh=\lf{i}t_{i}=\lf{0}t_{0}+\lf{1}t_{1}+\lf{2}t_{2}.
\end{equation}
They satisfy
\begin{equation}
d\lf{i}+\frac{1}{2}\varepsilon^{i}_{\;\;jk}\lf{j}\wedge\lf{k}=0. \label{frame structure equations}
\end{equation}
We make frequent use of  the complex combinations
\begin{equation}
\sigma=\lf{1}+i\lf{2}, \quad \bar{\sigma}=\lf{1}-i\lf{2}, \label{complex structure eq}
\end{equation}
for which we have that
\begin{equation}
d\sigma=-i\sigma\wedge\lf{0}, \quad d\lf{0}=\frac{i}{2}\bar{\sigma}\wedge \sigma.
\end{equation}
In terms of the complex coordinates we find that
\begin{equation}
\label{sigmaz}
\sigma=\lf{1}+i\lf{2}=2i(z_1dz_2-z_2dz_1),\quad \lf{0}=2i(\bar{z}_{2}dz_2-\bar{z}_{1}dz_1).
\end{equation}

The dual left-invariant vector fields $X_i$, $i=0,1,2$, generate the right action $h\to ht_{i}$ and satisfy
\begin{equation}
[X_{i},X_{j}]=\varepsilon_{ij}^{\;\;\;k}X_{k},
\end{equation}
so that the complex linear combinations, $X_{\pm}=X_{1}\pm iX_{2}$, satisfy
\begin{equation}
[X_{+},X_{-}]=-2iX_{0}, \quad [X_{0},X_{\pm}]=\pm iX_{\pm}.
\end{equation}
In terms of the complex coordinates they  are
\bee
X_{0}=\frac{i}{2}(z_2\partial_{2}+z_1\partial_{1}-\bar{z}_{2}\bar{\partial}_{2}-\bar{z}_{1}\bar{\partial}_{1}),\quad X_{-}=\bar{X}_{+}=X_{1}-iX_{2}=-i(\bar{z}_{1}\partial_{2}+\bar{z}_{2}\partial_{1}).
\eee
The only non-zero pairings are
\begin{equation}
\lf{0}(X_{0})=1, \quad \sigma(X_{-})=\bar{\sigma}(X_{+})=2. \label{pairings}
\end{equation} 

The Poincar\'e disk model of hyperbolic  two-space  can also be viewed as the coset space
\begin{equation}
H^{2}\simeq SU(1,1)/U(1),
\end{equation}
where we consider the $U(1)$ generated by $t_{0}$. As $H^{2}$ is contractible this means that $SU(1,1)$ is a trivial circle bundle over $H^{2}$, with $X_{0}$ generating translation in the fibre direction. 

We can construct a projection from the group manifold $SU(1,1)$ to $H^{2}$, in an analogous manner to the projection in the Hopf fibration \cite{RS}. Using the complex coordinate $z\in \C$ on $H^{2}$ and in terms of the complex coordinates $(z_1,z_2)$ for $SU(1,1)$ this projection is
\begin{equation}
\lh:SU(1,1)\to H^{2},\quad h\mapsto z=\frac{z_2}{z_1}. \label{ads-projection}
\end{equation}
A global section of this bundle is given by
\begin{equation}
s:H^{2}\to SU(1,1), \quad z\mapsto \frac{1}{\sqrt{1-|z|^{2}}}\begin{pmatrix}
1&\bar{z}\\
z&1
\end{pmatrix}. \label{ads-section}
\end{equation}

\subsection{Hyperbolic vortices as Cartan connections}
We now show that the hyperbolic vortex equations can be interpreted as the flatness conditions for a $su(1,1)$ Cartan connection which encodes the geometry, modified by the vortices.

\begin{prop}
\label{Cartanhyp}
The frame \eqref{complex frame} and the spin connection \eqref{spinconnection} for $H^{2}$ can be combined into the $su(1,1)$ gauge potential 
\begin{equation}
\label{hatA}
\hat{A}=\Gamma t_{0}+\frac{1}{2i}\left(et_{-}-\bar{e}t_{+}\right).
\end{equation}
The flatness condition for $\hat{A}$ is equivalent to the structure equation \eqref{structureeq} and Gauss equation \eqref{Gauss} on $H^{2}$. The flatness of the pull-back potential $f^{*}\hat{A}$, for holomorphic $f:H^{2}\to H^{2}$, is equivalent to the hyperbolic vortex equations being satisfied by the pair $(\phi,a)$ defined through \eqref{vortexpair}.
\end{prop}
In other words, the gauge potential $\hat{A}$ defines  a Cartan connection describing the hyperbolic geometry of the Poincar\'e disk model of two-dimensional hyperbolic space while $f^{*}\hat{A}$ is the gauge potential for  a Cartan connection describing the deformed geometry defined by the hyperbolic vortex $(\phi,a)$.
\begin{proof}
The curvature of $\hat{A}$ is
\bee
F_{\hat{A}}	=d\hat{A}+\frac{1}{2}[\hat{A},\hat{A}]
				=\left(\mathcal{R}+\frac{i}{2}e\wedge \bar{e}\right)t_{0}+\frac{1}{2i}\left(de-i\Gamma\wedge e \right)t_{-}-\frac{1}{2i}\left(d\bar{e}+i\Gamma\wedge\bar{e}\right)t_{+}.
\eee
The coefficient of $t_{0}$ being zero is equivalent to the Gauss equation \eqref{Gauss}, and the coefficients of $t_{\pm}$ being zero is equivalent to the structure equation \eqref{structureeq}. For $f^{*}\hat{A}$ we can use \eqref{vortexpair} to see that
\begin{equation}
f^{*}\hat{A}=(a+\Gamma)t_{0}+\frac{1}{2i}\left(\phi et_{-}-\bar{\phi}\bar{e}t_{+}\right), \label{pulledback connection}
\end{equation}
which has curvature
\begin{equation}
f^{*}F_{\hat{A}}=\left(da-(|\phi|^{2} -1)\mathcal{R}\right)t_{0}+\frac{1}{2i}\left(d\phi -ia\phi\right)\wedge et_{-}-\frac{1}{2i}\left(d\bar{\phi}+ia\bar{\phi}\right)\wedge\bar{e}t_{+}.
\end{equation}
The hyperbolic vortex equations \eqref{hyperbolic vortex} being satisfied is thus equivalent to the vanishing of this curvature.
\end{proof}

\section{Vortices on \texorpdfstring{$SU(1,1)$}{SU(1,1)} }

\subsection{Vortex equations and flat \texorpdfstring{$SU(1,1)$}{SU(1,1)} connections }
We now define and solve vortex equations on the group manifold of $SU(1,1)$ and show that  three-dimensional vortex configurations which solve them  are  equivariant  versions  of  hyperbolic vortices. 

To prepare for the equivariant description, we define the space of equivariant functions on $SU(1,1)$ as
\begin{equation}
C^{\infty}(SU(1,1),\C)_{N}=\{F:SU(1,1)\to \C| 2iX_{0}F=-NF\}, \quad N\in \Z, \label{equivariant degree def}
\end{equation}
with $N$ called the equivariant degree of the function.
One finds that
\begin{equation}
i((1-|z|^{2})\bar{\partial}-\frac{N}{2}z)(s^{*}F)(z)=s^{*}(X_{+}F) , \label{left-invariant to holomorphic}
\end{equation}
where $F\in C^{\infty}(SU(1,1),\C)_{N}$.
From this we deduce the following commutative diagram
\begin{equation}
\xymatrix{C^{\infty}(SU(1,1),\C)_{N}\ar[r]^{X_{+}}\ar[d]_{s^{*}}& C^{\infty}(SU(1,1),\C)_{N+2} \ar@{->}[d]^{s^{*}}\\  C^{\infty}(H^{2})\ar@{->}[r]_{i((1-|z|^{2})\bar{\partial}-\frac{N}{2}z)}& C^{\infty}(H^{2}).} \label{X plus holomorphicity}
\end{equation}
We will see later that functions with a well defined equivariant degree on $SU(1,1)$ can be used to construct the lift of a vortex of finite charge from $H^{2}$; it is these lifts of finite charge vortices that are the analogues of the vortex configurations considered in \cite{RS}.

\begin{dfn}
A  pair $(\Phi,A)$ of a one-form  $A$ on $SU(1,1)$ and  a map $\Phi:SU(1,1)\to \C$ is called  a vortex configuration
on $SU(1,1)$  if it  satisfies the 
vortex equations
\begin{equation}
 \left( d\Phi-iA\Phi\right)\wedge\sigma=0 ,  \qquad F_{A}=-\frac{i}{2}\left(|\Phi|^{2}-1\right)\sigma\wedge\bar{\sigma},
\label{ads vortex equations}
\end{equation}
where $F_{A}=dA$.
\end{dfn}

Note, that unlike in the Euclidean case considered in \cite{RS}, we have not imposed any normalisation condition on $A$ and have not fixed the equivariant degree of $\Phi$.  However, the vortex equations imply the following equivariance condition:
\begin{equation}
\mathcal{L}_{X_{0}}A=d (A(X_0)), \quad i\mathcal{L}_{X_{0}}\Phi=-A(X_{0})\Phi. \label{equivariance}
\end{equation}
The first follows from the Cartan formula $\mathcal{L}_{X_{0}}  = d\iota_{X_0} + \iota_{X_0} d$ and the form of $F_A=dA$ dictated by the second vortex equation. The  equivariance condition for $\Phi$  can be obtained  by contracting the first vortex equation with $(X_{0},X_{-})$.   We discuss the case of $\Phi$ having equivariant degree $2N-2$, the analogue of the spherical case \cite{RS},  later.

The vortex equations  \eqref{ads vortex equations} clearly resemble the hyperbolic vortex equations, \eqref{hyperbolic vortex}, with the complexified left-invariant one-forms, $\sigma$ and $ \bar{\sigma}$ replacing the complexified frame, $e$ and $\bar{e}$.  We will establish the precise relation between the two at the end of this section.

We now come to the central result of this section, which is a  three-dimensional analogue of 
the description of hyperbolic vortices in terms of a flat  $SU(1,1)$ connection given in \eqref{Cartanhyp}.  It is also the Lorentzian analogue of   Theorem 3.2 in \cite{RS}, but  differs from it in two important respects. In the Euclidean version, the relevant $U(1)$ bundle is the Hopf bundle, and associated line bundles are classified by an integer degree, but the total space $SU(2)$ is simply-connected. Here,  the $U(1)$ bundle is trivial, but the total space $SU(1,1)$ is not simply connected. The generator of the first fundamental form is the curve
\bee
\gamma=\{ e^{\varphi t_0} \in SU(1,1) | \varphi \in [0,4\pi)\},
\eee
which enters our condition for vortex configuration to be globally solvable. 
 
  \begin{thm}
 \label{vortex connection theorem}
 A vortex configuration on $SU(1,1)$ determines a gauge potential for a flat $SU(1,1)$ connection on $SU(1,1)$ 
through the following expression:
\begin{equation}
\mathcal{A}=(A+\lf{0})t_{0}+\frac{1}{2}\left(\Phi\sigma t_{-}+\bar{\Phi}\bar{\sigma}t_{+}\right). \label{A connection}
\end{equation}
Conversely, any flat $SU(1,1)$ connection $  \mathcal{A}$  on $SU(1,1)$ with 
\bee
\label{nonabcond}
\mathcal{A}(X_0)= p t_0, \quad\mathcal{A} (X_-) =\alpha t_0 + \Phi t_-,
\eee
for functions $p:SU(1,1)\rightarrow \R$ and $\alpha, \Phi: SU(1,1)\rightarrow \C$,  determines a vortex configuration  $(\Phi,A)$ via the expansion  \eqref{A connection}.  \\
 A gauge potential $\mathcal{A}$ for a flat $su(1,1)$ connection on $SU(1,1)$ of the form \eqref{A connection} which satisfies 
\bee
\label{intcond}
\int_\gamma \mathcal{A} =4\pi n t_0,
\eee
for some $n\in\Z$, 
can be trivialised as $\mathcal{A}=V^{-1}dV$, where $V:SU(1,1)\to SU(1,1)$ is a bundle map covering a holomorphic map $f:H^{2}\to H^{2}$.  Without loss of generality we can write  $V$ in the form
\begin{equation}
V:(z_1,z_2)\mapsto\frac{1}{\sqrt{|F_1|^{2}-|F_2|^{2}}}\begin{pmatrix}
F_1&\bar{F}_{2}\\
F_2& \bar{F}_{1}
\end{pmatrix}, \label{ads bundle map}
\end{equation}
where $F_1$ and $F_2$ are maps $SU(1,1) \rightarrow \C$,  with  $|F_1|>|F_2|$ and in particular  $F_1 \neq 0$. 
The vortex configuration $(\Phi, A)$ can be computed from the bundle map $V$ through
\begin{equation}
V^{*}\sigma=\Phi\sigma, \qquad A=V^{*}\lf{0}-\lf{0}. \label{vortex field defs}
\end{equation}
\end{thm}
\begin{proof}
The proof that the vortex equations \eqref{ads vortex equations} imply flatness of $\mathcal{A}$ given in \eqref{A connection} is a simple calculation, see also \cite{RS}. Conversely, expanding an  $su(1,1)$-valued one-form $\mathcal{A}$ on $SU(1,1)$  in terms of the generators $t_0,t_+$ and $t_-$, with coefficients which are linear combinations of the one-forms $\sigma^0,\sigma$ and $ \bar{\sigma}$,  and imposing \eqref{nonabcond} leads to a gauge potential $\mathcal{A}$  of the form \eqref{A connection}  with Higgs field $\Phi$ and abelian gauge field
\bee 
A = (p-1)\sigma^0  + \frac{1}{2}\left(\alpha \sigma + \bar{\alpha} \bar{\sigma}\right).
\eee
 The flatness of  $\mathcal{A}$ then give the vortex equations \eqref{ads vortex equations}, as already noted.

A  connection $\mathcal{A}$ on $SU(1,1)$ can be trivialised in terms of  $V:SU(1,1)\to SU(1,1)$ as $\mathcal{A}=V^{-1}dV$ if its path-ordered exponential is path-independent. In that case one can construct $V$ explicitly from the path-ordered exponential of $\mathcal{A}$ along any path, starting at a fixed (but arbitrary) base point,  see for example  \cite{BaezMuniain}. 

In our case, the  flatness of  $\mathcal{A}$ ensures the path-independence of the path-ordered exponential for all contractible paths on $SU(1,1)$, by the non-abelian Stokes Theorem. The condition \eqref{nonabcond} implies  that the path-ordered exponential and the exponential of the  ordinary integral of $\mathcal{A}$  along $\gamma$   coincide,  and  finally \eqref{intcond} ensures   that the   path-ordered exponential of $\mathcal{A}$ along $\gamma$ is trivial: 
\bee
\mathcal{P} \exp(\int_\gamma \mathcal{A}) = \mathbb{I}.
\eee
Again using  flatness of $\mathcal{A}$ we conclude  that the path-ordered exponential along any closed curve on $SU(1,1)$  is the identity, thus establishing the path-independence of the path-ordered exponential.

It remains  to show that the requirements \eqref{nonabcond} for   $\mathcal{A} =V^{-1}dV$ force $V$ to be a bundle map covering a holomorphic map $H^2\rightarrow H^2$.  The  first  condition becomes
\begin{equation}
X_{0}V =pV t_{0}, \label{bundlemap property}
\end{equation}
for $p: SU(1,1) \rightarrow \R$. However, this  is  precisely the infinitesimal formulation of the requirement that $V$ preserves the fibres of the fibration $SU(1,1) \rightarrow H^2$, i.e., that $V$ is a  bundle map. 

The second condition in \eqref{nonabcond} complex conjugates to 
\bee
\label{secondcond}
V^{-1} X_+V=\bar{\alpha} t_0 + \bar{\Phi} t_+.
\eee
Applying  
\begin{equation}
V^{-1}dV=V^{*}\lf{0} t_{0}+\frac{1}{2}\left(V^{*}\sigma t_{-}+V^{*}\bar{\sigma}t_{+}\right)
\end{equation} 
to $X_+$,  the condition \eqref{secondcond} is thus  equivalent to the vanishing of the $t_-$-component in $V^{-1}dV$: 
\bee
\label{pullbacksigma}
V^*\sigma (X_+) =0.
\eee
We need to show that this is equivalent to $V$ covering a holomorphic map.

Using the parameterisation \eqref{ads bundle map} of $V$, we see from \eqref{bundlemap property}  that the components of $V$ satisfy
\begin{equation}
X_{0}\left(\frac{F_{i}}{\sqrt{|F_1|^{2}-|F_2|^{2}}}\right)=\frac{i}{2}p \left(\frac{F_{i}}{\sqrt{|F_1|^{2}-|F_2|^{2}}}\right). \label{bundle map components}
\end{equation}
It follows that the map  $F=\pi\circ V = F_2/F_1$ has equivariant degree zero, and that $V$ covers the map
\begin{equation}
f=s^{*}\left(\frac{F_2}{F_1}\right):H^{2}\to H^{2}.
\end{equation}
Applying \eqref{left-invariant to holomorphic}  to the map $F=F_2/F_1$  we deduce that $f$ being holomorphic  is equivalent to 
 \bee 
X_{+}\left(\frac{F_2}{F_1}\right)=0.
\eee
However, again recalling that $F_1\neq 0$ and using  the explicit form \eqref{sigmaz} of $\sigma$,  this is seen to be  equivalent to the condition \eqref{pullbacksigma} on $V$. Thus, the requirement \eqref{nonabcond} forces $V$ to be a bundle map covering a holomorphic map, as claimed. 
\end{proof}

The Theorem and its proof deserve a few comments. 
First we note that the vortex  equations  \eqref{ads vortex equations}  are  invariant under  $U(1)$ gauge transformations  of the form
\begin{equation}
(\Phi,A)\mapsto(e^{i\beta}\Phi, A+d\beta), \quad \beta\in C^{\infty}(SU(1,1)). \label{su(1,1) U(1) gauge transform}
\end{equation}
The $U(1)$ gauge invariance is implemented at the level of the bundle map $V$ via 
\begin{equation}
V\mapsto \tilde V = Ve^{\beta t_{0}},\quad \beta\in C^{\infty}(SU(1,1)).
\end{equation}
This new trivialisation defines the same $f$ as $V$ and leads to a connection $\tilde{V}^{-1}d\tilde{V}$ differing  from $\mathcal{A}=V^{-1}dV$ by the $U(1)$ gauge transformation given in \eqref{su(1,1) U(1) gauge transform}.

Secondly, we observe that we can build vortex configurations on $SU(1,1)$ from a given  holomorphic map $f:H^{2}\to H^{2}$ by choosing  
\bee
\label{triviallift}
F_1(z_1,z_2)=1,\qquad 
F_2(z_1,z_2)=f\left(\frac{z_2}{z_1}\right).
\eee 
For this choice of $V$ the connection satisfies $\mathcal{A}(X_{0})=0$ and also, by flatness,   $\mathcal{L}_{X_0} \mathcal{A}=0$, so that $\mathcal{A}$ is constant along the fibre.

Finally, the condition \eqref{intcond} is needed to ensure the existence of a globally defined trivialisation of $\mathcal{A}$. When it is violated, one can still trivialise, but one will in general have to work with local  trivialisations, defined in at least two  simply connected patches which cover $SU(1,1)$. We will not consider such trivialisations in this paper, but they may well be of interest. 

\subsection{Vortex configurations of finite equivariant degree}

In the following  we exhibit a  choice of  bundle map which is a natural lift of the Blaschke product considered in Sect.~\ref{hyphol}. The resulting vortices  on $SU(1,1)$ are  lifts of hyperbolic vortices with a finite number of zeros, and  a Lorentzian analogue of the vortices obtained from homogeneous polynomials in the Euclidean version considered in \cite{RS}. 

Before stating our result we define  functions 
 $F_1, F_2:SU(1,1) \rightarrow \C $    for given complex numbers $c_k$, $k=1,\ldots, N$, in the unit disk  via
\begin{equation}
F_1=\prod_{k=1}^{N}(z_1-\bar{c}_kz_2), \qquad F_2=\prod_{k=1}^{N}(z_2-c_kz_1), \label{blaschke form}
\end{equation}
and note that $F_2/F_1$ is a function of $z=z_2/z_1$ and given by the Blaschke product \eqref{fblaschke}. In particular,  therefore  $|F_1|>|F_2|$, and we can use $F_1,F_2$ to define a bundle map $V$ covering the Blaschke product \eqref{fblaschke} via \eqref{ads bundle map}.

\begin{cor}
\label{finite charge vortices on H^2}
Let $V:SU(1,1) \rightarrow SU(1,1)$ be the bundle map \eqref{ads bundle map} constructed from \eqref{blaschke form}. 
Then the vortex configuration $(\Phi,A)$ constructed from the flat connection $\mathcal{A} =V^{-1} dV$ via Theorem \ref{vortex connection theorem} has a Higgs field of equivariant degree $2N-2$ and a gauge field which  satisfies the normalisation condition 
\bee
A(X_{0})=N-1.
\eee
The  vortex configuration $(\Phi,A)$ can be given in terms of $F_1, F_2$ as
\begin{equation}
\Phi=\frac{F_1\partial_{2}F_2-F_2\partial_{2}F_1}{z_1(|F_1|^{2}-|F_2|^{2})}, \label{explicit Higgs field}
\end{equation}
and 
\begin{equation}
A=(N-1)\lf{0}-\frac{i}{2}X_{-}\ln D^{2} \sigma+\frac{i}{2}X_{+}\ln D^{2}\bar{\sigma}, \label{vortex gauge potential}
\end{equation}
where $D^{2}=|F_1|^{2}-|F_2|^{2}$.
\end{cor}
\begin{proof}
It follows from the explicit form of $F_1,F_2$ that
\bee
X_0V=NVt_0,
\eee
so that $\mathcal{A}(X_0) = Nt_0$ and therefore, by the decomposition \eqref{A connection},   $A(X_{0})=N-1$. The general equivariance condition \eqref{equivariance} then implies 
\bee
i\mathcal{L}_{X_0} \Phi = -(N-1)\Phi,
\eee
so that $\Phi$ has equivariant degree $2N-2$ by \eqref{equivariant degree def}. To get the explicit expression for $\Phi$  we compute
\begin{equation}
\Phi=\frac{1}{2}V^{*}\sigma(X_{-}).
\end{equation}
 To compute $A$ we use 
\begin{equation}
V^{*}\lf{0}(X_{0})=N, \quad V^{*}\lf{0}(X_{+})=iX_{+}\ln D^{2}, \quad V^{*}\lf{0}(X_{-})=-iX_{-}\ln D^{2}, 
\end{equation}
to get the claimed result.
\end{proof}

Recall that the Blaschke product \eqref{fblaschke} gives rise, via the pull-back construction,  to a hyperbolic vortex of charge $N-1$ on $H^2$. The Corollary above shows that the natural covering $V$ of the Blaschke product defines a vortex configuration on $SU(1,1)$ of  finite equivariant degree $2N-2$.

For finite Blaschke products, the lift \eqref{blaschke form} defines a natural bundle map which we  used to construct a three-dimensional vortex configuration with non-zero degree. In general,
lifting to a configuration with non-zero equivariant degree is not trivial. For example,  the  vortices on the hyperbolic cylinder  studied   in \cite{MR,MM} require the infinite Blaschke product  
\begin{equation}
f(z)=z\prod_{j=1}^{\infty}\left(\frac{z-a_{j}^{2}}{1-a_{j}^{2}z}\right)^{2},
\end{equation}  
where the zeros of $f$ are at $a_{j}=i\tanh\left(\frac{j\lambda}{2}\right)$ with $\lambda$ defined in \cite{MR,MM} as $\lambda=\frac{\pi K^{\prime}(k)}{K(k)}$ for $K$ the elliptic integral of the first kind,   $K^{\prime}(k)=K(\sqrt{1-k^2})$ and any $0<k<1$.
These can still be lifted to $SU(1,1)$ via the  lift \eqref{triviallift}, but there does not appear to be any non-trivial natural option. 

\subsection{Lifting Cartan connections  for  hyperbolic vortices}
We have already seen how to lift  hyperbolic vortices to  vortex configurations on $SU(1,1)$. Since the latter can be expressed in terms of a flat $SU(1,1)$ connection, it is natural to expect a link with the Cartan connection encoding hyperbolic vortices according to Proposition \ref{Cartanhyp}.  In this short section, we exhibit this link.

For a nowhere-vanishing  function $g: H^2 \to \C$ define the map
\begin{equation}
r_g: H^2 \to SU(1,1) , \quad r_g=\begin{pmatrix}
\frac{\bar{g}}{|g|}&0\\
0&\frac{g}{|g|}
\end{pmatrix}. \label{rmap}
\end{equation}
We use  this map as a gauge transform in the following Lemma, which is a Lorentzian analogue of  Lemma 4.3  in \cite{RS}.
\begin{lem} \label{lifted connection lemma}
With the section $s:H^2\rightarrow SU(1,1)$ defined as in \eqref{ads-section}, the gauge potential \eqref{hatA} for the Cartan connection of the hyperbolic disk is trivialised by $s$:
\begin{equation}
\hat{A}=s^{-1}ds. \label{H^2 cartan connection}
\end{equation}
If $V$ is a bundle map of the form \eqref{ads bundle map} covering  a holomorphic map $f:H^{2}\to H^{2}$,  then the gauge potential $f^{*}\hat{A}$ for the deformed Cartan geometry and the pull-back via $s$ of $\mathcal{A}=V^{-1}dV$ are related through the gauge transformation $r_{f_{1}}$, where $f_1=F_1\circ s$:
\begin{equation}
f^{*}\hat{A}=r^{-1}_{f_{1}}s^{*}(\mathcal{A})r_{f_{1}}+r^{-1}_{f_{1}}dr_{f_{1}}. \label{hyperbolic vortex vortex config relation}
\end{equation} 
\end{lem}
Note that,  if we work in the gauge where $F_1=1$ and $F_2=f\left(\frac{z_2}{z_1}\right)$, then $r_{f_{1}}=\mathbb{I}$ , so $f^{*}\hat{A}$ and $s^{*}\left(V^{-1}dV\right)$ agree. More generally, the fact that $F_1$ and therefore $f_{1}$ has no zeros   means that the   gauge transformation $r_{f_{1}}$ is  smooth. This is in contrast to the Euclidean case considered in \cite{RS} where a singular gauge transformation was needed.
\begin{proof}
The proof is a straightforward  calculation which proceeds along the lines given in \cite{RS}. To show \eqref{H^2 cartan connection} one uses  \eqref{ads-section} and compares to the definitions of $e$ and $\Gamma$ in terms of the complex coordinates in \eqref{complex frame} and \eqref{spinconnection}. To show \eqref{hyperbolic vortex vortex config relation}, one notes  $ 
s\circ f =(V\circ s)r_{f_{1}} $ 
and $\mathcal{A}= V^{-1} d V$. 
\end{proof}

\section{Magnetic Dirac operators on \texorpdfstring{$\AAdS$}{AdS3} and  Minkowski space}

\subsection{Notational conventions}
We denote three-dimensional Minkowski space by $\R^{1,2}$,  and use a `mostly minus' Lorentzian metric $\eta$ with matrix \eqref{etadef} in an orthonormal basis. We write elements  of Minkowski space as $\vec{x}= (x^0,x^1,x^2)^t$ so that 
 \bee
 \eta= (dx^{0})^{2}-(dx^{1})^{2}-(dx^{2})^{2}.
 \eee
 Our volume element is  $dx^{0}\wedge dx^{1}\wedge dx^{2}$  so that  $dx^0,dx^1,dx^2$   is an oriented basis of the cotangent space.
We  use indices $i,j\ldots$  in the range $0,1,2$, raised  and lowered using $\eta_{ij}$. The scalar and vector product are given by
\bee
\vec{x}\cdot \vec{y}=x^{i}y_{i}, \qquad 
(\vec{x}\times \vec{y})^k=\varepsilon_{ij}^{\;\;\;k}x^{i}y^{j},
\eee
for $\vec{x},\vec{y}\in \R^{1,2}$, $\varepsilon_{012}=1$  and the summation convention being understood between pairs of raised and lowered indices. The Lorentzian length squared  of  $\vec{x}$ is denoted  by
\begin{equation}
r^{2}=x^{i}x_{i}=(x^{0})^{2}-(x^{1})^{2}-(x^{2})^{2}.
\end{equation} 
We also use the notation  $ \partial_i=\partial / \partial x^i $ for partial derivatives.

On $SU(1,1)$ we continue to  work with the notation introduced in Sect.~\ref{su11conventions}, and   use the following oriented orthonormal frame consisting of 
\begin{equation}
 \frac{1}{2}\lf{0},\frac{1}{2}\lf{1}, \frac{1}{2}\lf{2},
\end{equation}
the metric
\begin{equation}
ds^{2}=\frac{1}{4}\left((\lf{0})^{2}-(\lf{1})^{2}-(\lf{2})^{2}\right) \label{ads metric},
\end{equation}
and orientation
\begin{equation}
\text{Vol}_{\text{ \tiny AdS}}= \frac{1}{8}\lf{0}\wedge\lf{1}\wedge\lf{2}. \label{ads orientation}
\end{equation}

Differential forms provide the natural language for our discussion, but occasionally we use the isomorphisms between  forms and vector fields which are possible on a three-dimensional manifold with a non-degenerate inner product and volume form. The inner product allows one to identify vector fields with one-forms; denoting  the volume form $\text{Vol}$, it   establishes a bijection between a  vector field $X$  and  a two-form $F$ via 
\bee
\label{2formvfield}
\iota_X\text{Vol} = F.
\eee
On $SU(1,1)$, for example, the vector field $X_0$ generating the fibre translation is mapped to the two-form $\frac  1 8 \sigma^1\wedge \sigma^2$ via \eqref{ads orientation}, and this will play a role in our discussions. 

In the following we will be considering the Dirac equation on both $SU(1,1)$ and $\R^{1,2}$,  so we need to fix our conventions for the Clifford algebra $Cl(1,2)$. The algebra is generated by the gamma matrices, $\gamma_{i}$ which satisfy
\begin{equation}
\{\gamma_{i},\gamma_{j}\}=-2\eta_{ij}.
\end{equation}
We pick $\gamma_{i}=2t_{i}$ as our representation.

\subsection{Stereographic projection and frames}
We will be using a stereographic projection to relate  Dirac operators on  $\AAdS\simeq SU(1,1)$ to Dirac operators on Minkowski space. We now set up our conventions and explain how  orthonormal frames on these spaces are mapped into each other via stereographic projection. 

To discuss stereographic projection from  $\AAdS$ to $\R^{1,2}$,  it is helpful to think of $\AAdS$ as a real manifold. 
 As a subspace of $\R^{2,2}$,   and  with AdS length $\ell$, it is given by
\begin{equation}
\label{AdSinR4}
\AAdS=\{(y^{0},y^{1},y^{2},y^{3})\in \R^{2,2}| (y^{0})^{2}-(y^{1})^{2}-(y^{2})^{2}+(y^{3})^{2}=\ell^{2}\}.
\end{equation}
This is just the definition of $\AAdS$ as a submanifold of $\C^{1,1}$ from \eqref{Ads cdef} written in terms of real coordinates. The real coordinates are related to the complex coordinates for $SU(1,1)$ \eqref{complex coordinate matrix} through 
\begin{equation}
\ell(z_1,z_2)=(y^{3}+iy^{0},y^{2}-iy^{1}). \label{real complex coordinate relation}
\end{equation}

Just as stereographic projection from the sphere requires one to single out a north and a south pole, we need to pick special points on $\AAdS$ to define our stereographic projection. We  choose
\bee
\label{Pdef}
P_\pm= (0,0,0,\pm \ell)\in \AAdS.
\eee 
Then, to map a point  $(y^{0},y^{1},y^{2},y^{3})\in \AAdS$ into  $\R^{1,2}$,  we draw the line between it and the point $P_-$; the image is the  intersection  with $\R^{1,2}$  at $\{y^{3}=0\}$, see Fig.~\ref{ads proj figure}. This intersection   does not exist  for all points in  $\AAdS$ but only for the subset
\begin{equation}
\widetilde{\text{AdS}}_{+}=\{(y^{0},y^{1},y^{2},y^{3})\in \AAdS | y^{3}>-\ell\}. \label{stereo domain of definition}
\end{equation} 
Moreover, the point of intersection necessarily   lies in the subset $\Il\subset \R^{1,2}$ defined as 
\begin{equation}
\Il=\{(x^{0},x^{1},x^{2})\in \R^{3}|r^{2}=(x^{0})^{2}-(x^{1})^{2}-(x^{2})^{2}>-\ell^{2}\}. \label{ads stereo subset}
\end{equation}
Geometrically $\Il$ is the inside of a single sheeted hyperboloid. 

Thus we can define the  stereographic projection map
\begin{equation}
\text{St}:\widetilde{\text{AdS}}_{+}\to \R^{1,2}, \quad (y^{0},y^{1},y^{2},y^{3})\mapsto \vec{x}=\left(\frac{\ell y^{0}}{\ell+y^{3}},\frac{\ell y^{1}}{\ell+y^{3}},\frac{\ell y^{2}}{\ell+y^{3}}\right)
\end{equation}
and its inverse
\begin{equation}
\text{St}^{-1}:\Il \in \R^{1,2}\to \widetilde{\text{AdS}}_{3}, \quad \vec{x} \mapsto (y^{0},y^{1},y^{2},y^{3})=\frac{\ell}{\ell^{2}+r^{2}}\left(2\ell x^{0},2\ell x^{1}, 2\ell x^{2},\ell^{2}-r^{2}\right).
\end{equation}

\begin{figure}[!htbp]
\centering
\includegraphics[width=7truecm]{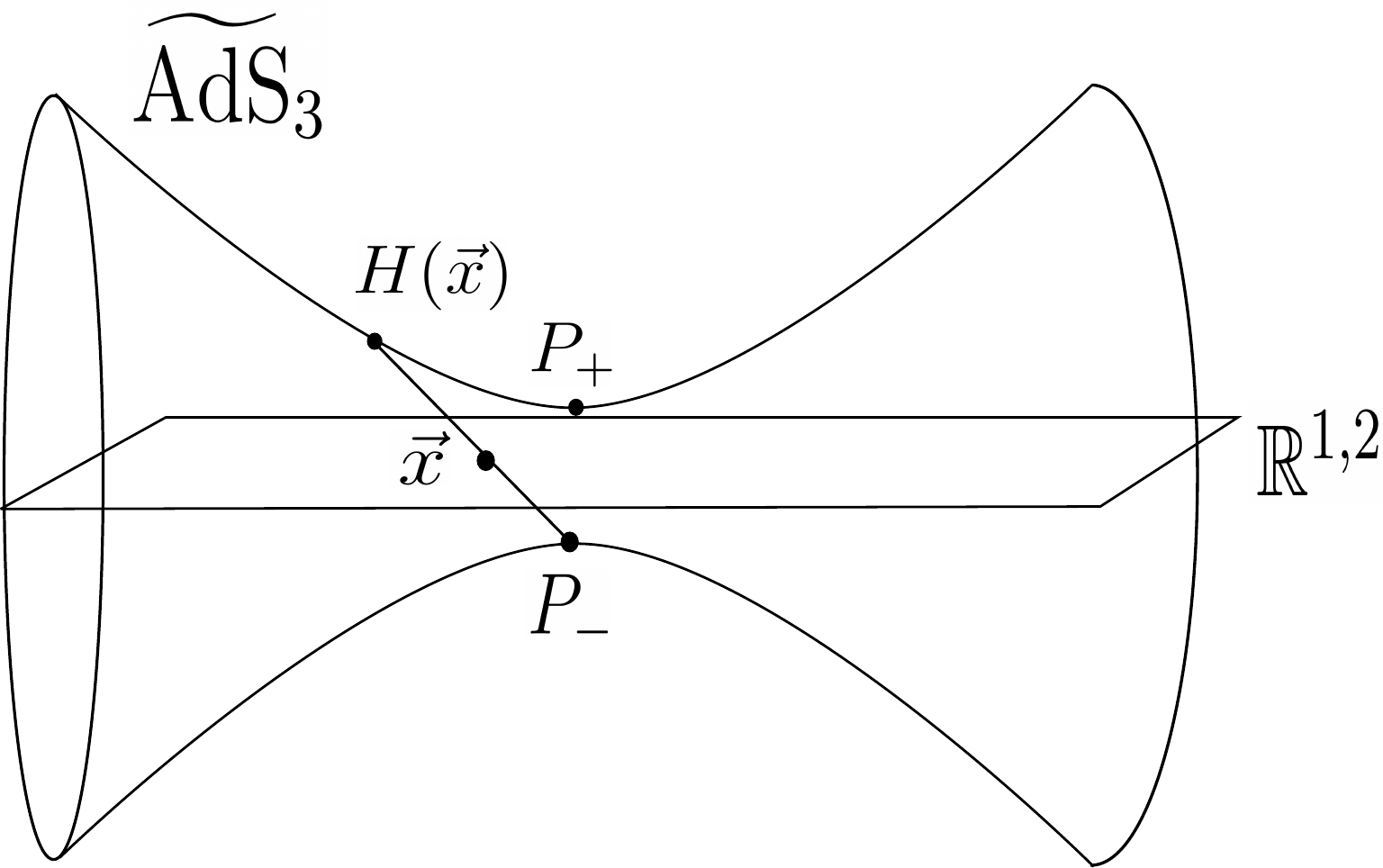}
 \caption{A schematic  picture of the stereographic projection from $\AAdS$ to $\R^{1,2}$, with one dimension suppressed;  we used the notation introduced in  \eqref{Pdef}}
 \label{ads proj figure}
\end{figure}

For calculations, it is helpful to express  some of these maps in matrix notation.
Using $
\vec{t}=(t^{0},t^{1},t^{2})^{t} $, we  
identify the point $(y^{0},y^{1},y^{2},y^{3})\in \AAdS$ with the $SU(1,1)$ matrix
\begin{equation}
M(y^{0},y^{1},y^{2},y^{3})=\frac 1 \ell (y^{3}\mathbb{I}+2\vec{y}\cdot\vec{t}). \label{SU(1,1) matrix from R4}
\end{equation}
Up to scale,  the inverse stereographic projection can then be written as 
\begin{align}
H:	\Il \subset\R^{1,2}&\to SU(1,1), \nonumber\\ 
	\vec{x}&\mapsto \frac{\ell^{2}-r^{2}}{\ell^{2}+r^{2}}\mathbb{I}+\frac{4\ell}{\ell^{2}+r^{2}}\vec{x}\cdot\vec{t} =\frac{1}{\ell^{2}+r^{2}}\begin{pmatrix}
	\ell^{2}-r^{2} +2i\ell x^{0}&2i\ell (x^{1}-ix^{2}) \\
	 -2i\ell(x^{1}+ix^{2})&  \ell^{2}-r^{2}-2i\ell x^{0}
	\end{pmatrix}. \label{scaled stereo}
\end{align}
In stereographic coordinates,  the  bundle projection $\pi:SU(1,1) \rightarrow H^2$  therefore becomes
\begin{equation}
 \lh \circ H : \vec{x} \mapsto \frac{z_2}{z_1}=-i\frac{2\ell(x^{1}+ix^{2})}{\ell^{2}-r^{2}+2i\ell x^{0}}.
\end{equation}
We  note that, in terms of \eqref{SU(1,1) matrix from R4}, the  condition $y^{3}>-\ell$ can be written as  $
\text{tr}(M)>-2$, and that matrices  which satisfy this lie in the image of the exponential map.

We will also  need a  Lorentzian version of the so-called gnomonic projection discussed  and used in \cite{RS}. This is the  map $G:\Il\subset \R^{1,2}\to SU(1,1)$ given by
\begin{equation}
G: \vec{x} \mapsto  \frac{1}{\sqrt{\ell^{2}+r^{2}}}\left(\ell\mathbb{I}+2\vec{x}\cdot\vec{t}\right), \label{gnomonic matrix}
\end{equation}
which satisfies $G^{2}=H$. 
The geometric interpretation of this result is shown  in Fig.~\ref{poincare_beltrami_maps2} and explained in the caption.
Note also   the map $H$ is an AdS version of  the  projection relating a sheet of the two-sheeted hyperboloid  to the disk in the Poincar\'e model. The map $G$ is the AdS analogue of  the map to the Beltrami disk model. 

\begin{figure}[t]
\centering
  \includegraphics[width=7truecm]{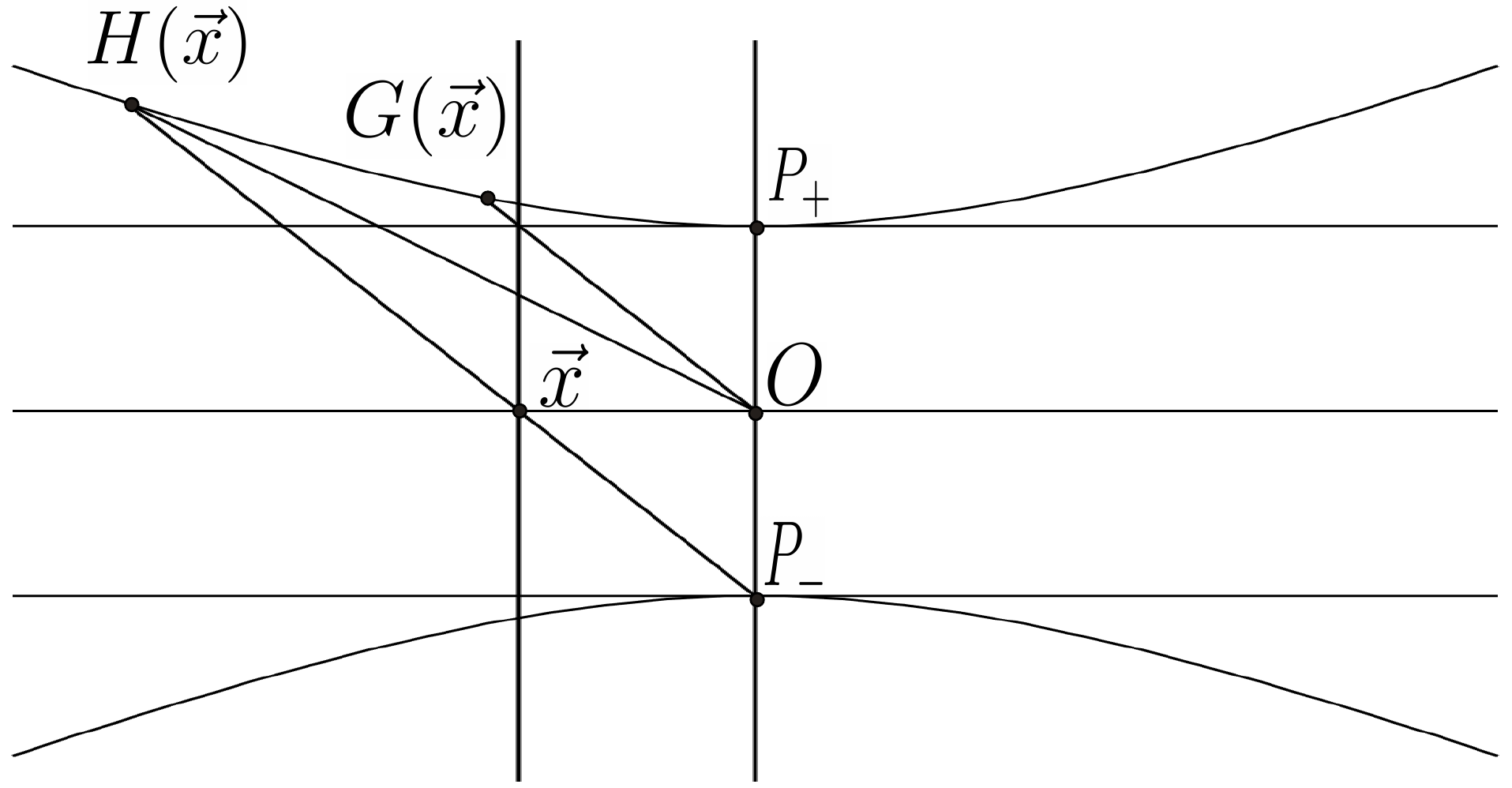}
      \caption{The  lines  $OG(\vec{x})$ and $P_-H(\vec{x})$  define  the maps $G$  and $H$. They are analogous  to, respectively, the   Beltrami  and Poincar\'e map in  hyperbolic geometry.  The area bounded by the hyperbolic segment $P_+H(\vec{x})$ and the straight lines $OH(\vec{x})$ and $OP_+$  is twice that  bounded by the hyperbolic segment $ P_+G(\vec{x})$ and the lines     $OG(\vec{x})$ and $OP_+$. This is the geometry behind the relation $H(\vec{x})=G^2(\vec{x})$.}
 \label{poincare_beltrami_maps2}
\end{figure}

By expanding $H^{-1}dH$ we define one-forms, $\vartheta_{i}$ $i=0,1,2$, on $\Il\subset \R^{1,2}$ via
\begin{equation}
H^{-1}dH=\frac{1}{\Omega}\vec{\vartheta}\cdot \vec{t}, \quad \vartheta^{i}=\Omega H^{*}\lf{i},
\end{equation}
where $\vec{\vartheta}=(\vartheta^{0},\vartheta^{1},\vartheta^{2})^t$ and we used the  scale factor 
\begin{equation}
\Omega=\frac{\ell^{2}+r^{2}}{4\ell}.
\end{equation}
We find that
\begin{equation}
\vec{\vartheta}\cdot \vec{t}= \frac{1}{\ell^{2}+r^{2}}\left(2(\vec{x}\cdot \vec{t})(\vec{x}\cdot d\vec{x})+(\ell^{2}-r^{2})d\vec{x}\cdot \vec{t}-2\ell(\vec{x}\times d\vec{x})\cdot \vec{t}\right)=G^{-1}(d\vec{x}\cdot \vec{t}) G.
\end{equation}
This means that the $\vartheta^{i}$, $i=0,1,2$,  give a Lorentz-rotated basis for the cotangent space of $\Il$.
\begin{lem}
The pullbacks of the Maurer-Cartan one-form via $G$ and $H $ are related through
\begin{equation}
H^{-1}dH=G^{-1}dG+G^{-1}(G^{-1}dG)G, \label{HG}
\end{equation}
with the inverse relation expressed as 
\begin{equation}
G^{-1}dG=\frac{1}{2}H^{-1}dH-\star\left(d\Omega\wedge H^{-1}dH\right).\label{GH}
\end{equation}
\end{lem}
\begin{proof}
The proof follows by a calculation which is similar to the corresponding  Euclidean version in  \cite{RS}, but differs  in  important signs.
The first statement \eqref{HG} follows from the fact that $H=G^{2}$. For \eqref{GH} we use \eqref{HG} to write it as
\begin{equation}
\star\left(2d\Omega\wedge\left(dGG^{-1}+G^{-1}dG\right)\right)=-\left(dGG^{-1}-G^{-1}dG\right). \label{GH midway}
\end{equation}
Then we compute
\begin{equation}
G^{-1}dG=\frac{2\left(\ell d\vec{x}\cdot \vec{t}-(\vec{x}\times d\vec{x})\cdot \vec{t}\right)}{\ell^{2}+r^{2}},
\end{equation}
which gives us that
\begin{equation}
dGG^{-1}+G^{-1}dG=\frac{4\left(\ell d\vec{x}\cdot \vec{t}\right)}{\ell^{2}+r^{2}}, \qquad dGG^{-1}-G^{-1}dG=\frac{4(\vec{x}\times d\vec{x})\cdot \vec{t}}{\ell^{2}+r^{2}}. \label{GdG inverse }
\end{equation}
With $2d\Omega=\frac{\vec{x}\cdot d\vec{x}}{\ell}$ and
\begin{equation}
\star\left(\vec{x}\cdot d\vec{x}\wedge d\vec{x}\cdot \vec{t}\right)=(d\vec{x}\times \vec{x})\cdot \vec{t}
\end{equation}
we get \eqref{GH midway} and hence \eqref{GH}.
\end{proof}

\subsection{Magnetic Dirac operators}
On $SU(1,1)$, a global gauge potential for the spin connection is given by
\begin{equation}
\Gamma_{SU(1,1)}=-\frac{1}{8}[\gamma_{i},\gamma_{i}]\omega^{ij}=\frac{1}{2}h^{-1}dh.
\end{equation} 
The Dirac operator in the left-invariant frame is then 
\begin{align}
\slashed{D}_{SU(1,1)}=\frac{4}{\ell}t^{i}X_{i}-\frac{3}{2\ell}\mathbb{I} =\frac{2i}{\ell}
\begin{pmatrix}
						X_{0}&-X_{-}\\
						X_{+}&-X_{0}
						\end{pmatrix} -\frac{3}{2\ell}\mathbb{I}.
\end{align}
Minimal  coupling to an abelian gauge potential $A$ yields
\begin{align}
\label{DiraconAdS}
\slashed{D}_{SU(1,1),A}	=\frac{4}{\ell}t^{i}(X_{i}+iA_{i})-\frac{3}{2\ell}\mathbb{I}
						=\frac{2i}{\ell}\begin{pmatrix}
						X_{0}+iA_{0}&-(X_{-}+iA_{-})\\
						X_{+}+iA_{+}&-(X_{0}+iA_{0})
						\end{pmatrix} -\frac{3}{2\ell}\mathbb{I}.
\end{align}

The  Dirac operator on $\R^{1,2}$  minimally coupled to $\vec{A}\cdot d\vec{x}$ is 
\begin{equation}
\slashed{D}_{\R^{1,2},A}=2t^{i}(\partial_{i}+iA_{i}). \label{mixed sig dirac}
\end{equation}

In the following,  spinors $\Psi$   which satisfy the massless Dirac equation $\slashed{D}_A\Psi =0$  on either $SU(1,1)$ or $\R^{1,2}$ coupled to an abelian gauge potential are called  magnetic Dirac  modes, or simply magnetic modes. The following  Lemma exhibits the relation between magnetic Dirac modes on  $SU(1,1)$ and  $\R^{1,2}$.
\begin{lem}\label{zero-modes relation}
If $\Psi:SU(1,1)\to \C^{1,1}$ is a magnetic mode of the Dirac operator on $SU(1,1)$ coupled to the $U(1)$ gauge field $A$ then
\begin{equation}
\Psi_{H}=G\Omega^{-1}H^{*}\Psi
\end{equation}
is a magnetic mode of the Dirac operator \eqref{mixed sig dirac} on $\Il\subset\R^{1,2}$ coupled to the gauge potential $H^{*}A$.
\end{lem}
\begin{proof}
The pull-back of the spin connection is
\begin{equation}
H^{*}\Gamma_{SU(1,1)}=\frac{1}{2}H^{-1}dH.
\end{equation}
Using \eqref{HG} one can show that
\begin{equation}
d+\frac{1}{2}H^{-1}dH=\Omega G^{-1}\left(d+\frac{1}{2}\left(GdG^{-1}+G^{-1}dG\right) +\Omega^{-1}d\Omega \right)\Omega^{-1} G.
\end{equation}
Then combining \eqref{GdG inverse } with
\begin{equation}
\Omega^{-1}d\Omega=\frac{2\vec{x}\cdot d\vec{x}}{\ell^{2}+r^{2}},
\end{equation}
gives that
\begin{equation}
t^{i}\iota_{\partial_{i}}(\frac{1}{2}\left(G dG^{-1}+G^{-1}dG\right)+\Omega^{-1}d\Omega)=-\frac{2\vec{x}\cdot \vec{t}}{\ell^{2}+r^{2}}+\frac{2\vec{x}\cdot\vec{t}}{\ell^{2}+r^{2}}=0.
\end{equation}
Using these results we compute  the pull-back of the Dirac operator on $SU(1,1)$, coupled to both the spin connection and the abelian gauge potential $A$, to the flat frame in $\R^{1,2}$:
\begin{align}
\frac{\ell}{2}H^{*}\slashed{D}_{SU(1,1),A}	&=2t^{i}\iota_{H^{*}X_{i}}\left(d+\frac{1}{2}H^{-1}dH+iH^{*}A\right),\nonumber\\								&=\Omega^{-1}G^{-1}2t^{i}\iota_{\partial_{i}}G\left(d+\frac{1}{2}H^{-1}dH+iH^{*}A\right),\nonumber\\										&=G^{-1}2t^{i}\iota_{\partial_{i}}\left(d+iH^{*}A\right)\Omega^{-1}G.
\end{align}
This implies the claimed relation between the  magnetic modes of $\slashed{D}_{SU(1,1),A}$ and $\slashed{D}_{\R^{1,2},H^{*}A}$.
\end{proof}

\section{Magnetic Dirac modes from vortices}

\subsection{Dirac modes on \texorpdfstring{$SU(1,1)$}{SU(1,1)} } 

We now show how to obtain  magnetic Dirac modes on $SU(1,1)$ from  vortex configurations on $SU(1,1)$. As a warm-up, we consider a simpler construction of magnetic modes  from a holomorphic function $F:SU(1,1) \rightarrow \C$. 

\begin{prop} \label{linear magnetic modes}
Let $n\in \N$ and consider  the gauge potential 
\begin{equation}
A=-\frac{2n+1}{4}\lf{0}
\end{equation}
and the  homogeneous  function $F_{n}=\sum_{k=0}^{n-1}a_{k}z_{2}^{k}z_{1}^{n-1-k}$ on $SU(1,1)$.
Then the  spinor
\begin{equation}
\Psi=\begin{pmatrix}
F_{n}\\0
\end{pmatrix}
\end{equation}
solves the Dirac equation on $SU(1,1)$ minimally coupled to $A$.
\end{prop}

\begin{proof}
First observe that 
\begin{equation}
2iX_{0}F_{n}=(1-n)F_{n},
\end{equation}
showing that  $F_{n}$ has equivariant degree  $n-1$.  Next note that $X_{+}F_{n}=0$ since $F_{n}$ is holomorphic, that $A(X_+)=0$
and that 
\begin{equation}
(X_{0}+iA(X_0))F_{n}=-\frac{3i}{4}F_{n}.
\end{equation}
Using this, and the explicit form of  $\slashed{D}_{SU(1,1),A}$ given in \eqref{DiraconAdS}, the  equation $\slashed{D}_{SU(1,1),A}\Psi =0$ 
   reduces to 
\begin{equation}
\frac{2i}{\ell}(X_{0}+iA_{0})F_{n}-\frac{3}{2\ell}F_{n}=0
\end{equation}
so $\Psi$ is indeed a magnetic mode.
\end{proof}

The following  Definition and Theorem are  similar to the  Euclidean version considered in  \cite{RS}. However, the non-linear equation in the definition of a vortex magnetic mode has an important overall sign difference.
\begin{dfn}
A pair $(\Psi, A)$ of a spinor $\Psi$ and a one-form $A$ on $SU(1,1)$ is said to be a vortex magnetic mode of the Dirac equation on $SU(1,1)$ if
\begin{equation}
\label{magneticmodedef}
\slashed{D}_{SU(1,1),A}\Psi=0, \quad F_{A}=-\frac{4i}{\ell}\star\Psi^{\dagger}h^{-1}dh\Psi -\frac{1}{4}\lf{1}\wedge\lf{2},
\end{equation}
with $\star$ the Hodge star operator on $SU(1,1)$ with respect to the metric \eqref{ads metric} and orientation \eqref{ads orientation}.
\end{dfn}
We now give a result that enables the construction of  magnetic modes from any vortex configuration.
\begin{thm} \label{vortex config modes}
Let $(\Phi, A)$ be a vortex configuration on $SU(1,1)$. Then the pair $(\Psi,A')$, where
\begin{equation}
\Psi=\begin{pmatrix}
\Phi\\0
\end{pmatrix}, \quad A'=-A-\frac{3}{4}\lf{0}, \label{magnetic mode from vortex}
\end{equation}
is a vortex magnetic mode on $SU(1,1)$.
\end{thm}
\begin{proof}
The spinor is  a magnetic mode of $\slashed{D}_{SU(1,1),A'}$ if
\begin{equation}
\left(iX_{0}-A'-\frac{3}{4} \right)\Phi=0 \quad \text{ and } \quad X_{+}\Phi+iA'_{+}\Phi=0.
\end{equation}
Now $A'_{0}=A'(X_{0})=-A(X_{0})-\frac{3}{4}$ so the first equation follows from the equivariance condition \eqref{equivariance}. The second follows from \eqref{ads vortex equations} since $A'(X_{+})=-A(X_{+})$ and contracting \eqref{ads vortex equations} with $(X_{+},X_{-})$ leads to 
\begin{equation}
X_{+}\Phi-iA(X_{+})\Phi=0.
\end{equation}
For the non-linear equation with a spinor of the form given in \eqref{magnetic mode from vortex} we have that
\begin{equation}
\frac{4i}{\ell}\star\Psi^{\dagger}h^{-1}dh\Psi=\frac{4i}{\ell}\star|\Phi|^{2}\left(\frac{i}{2}\lf{0}\right)=-|\Phi|^{2}\lf{1}\wedge\lf{2}.
\end{equation}
On the other hand using \eqref{ads vortex equations} gives
\begin{equation}
F_{A'}=-F_{A}+\frac{3}{4}\lf{1}\wedge\lf{2}=\left(|\Phi|^{2}-\frac{1}{4}\right)\lf{1}\wedge\lf{2},
\end{equation}
from which the non-linear equation follows.
\end{proof}

Both the magnetic two-forms $F_{A'}$ and $F_{A}$ are proportional  to    $\lf{1}\wedge\lf{2}$, with a factor of proportionality  which is a function on  $SU(1,1)$.   The magnetic vector field associated to either of them via \eqref{2formvfield} is therefore similarly proportional to  the vector field $X_0$,   and so the magnetic fields are just  the fibres of the fibration $\lh:SU(1,1)\to H^{2}$. 

To  visualise these fibres in the embedding of $\AAdS$ in $\R^{2,2}$ (with one dimension suppressed), we  note that they are in particular geodesics on $\AAdS$ and therefore can be obtained by intersections of  the embedding \eqref{AdSinR4} with planes in $\R^{2,2}$. This was used to produce the picture of geodesics on AdS$_2$, embedded in  $\R^{2,1}$, in Fig.~\ref{geodesics figure}.

 We can write  down  geodesics  on $\AAdS$ explicitly by expressing the right action  of  $e^{\alpha t_{0}}$, $\alpha \in [0,4\pi)$, which generates them,  in  real coordinates.   Using the parametrisation \eqref{SU(1,1) matrix from R4} the orbit of a point 
$(y^{0},y^{1},y^{2},y^{3})\in \AAdS $ is 
\begin{equation}
\left(\left(y^{0}\cos\frac{\alpha}{2}+y^{3}\sin\frac{\alpha}{2}\right),\left(y^{1}\cos\frac{\alpha}{2}-y^{2}\sin\frac{\alpha}{2}\right),\left(y^{2}\cos\frac{\alpha}{2}+y^{1}\sin\frac{\alpha}{2}\right),\left(y^{3}\cos\frac{\alpha}{2}-y^{0}\sin\frac{\alpha}{2}\right)\right).
\end{equation}
In other words moving along the fibre is equivalent to a rotation by $\frac{\alpha}{2}$ in both the $y^{3},y^{0}$ and $y^{1},y^{2}$ plane. 
 
 \begin{figure}[!htbp]
\centering
\includegraphics[width=10truecm]{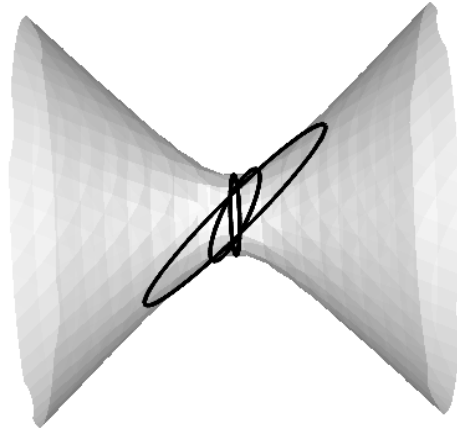}
\vspace{-4cm}
 \caption{Some geodesics on AdS$_2$}
 \label{geodesics figure}
\end{figure}

\subsection{Dirac modes on Minkowski space}
 In \cite{LY}, Loss and Yau used a particular formula to construct gauge potentials for a given spinor so that the spinor is a zero-mode of the Dirac  operator coupled to the gauge potential. This formula has a simple Lorentzian analogue, namely 
\begin{equation}
A_{i}=\frac{1}{|\vec{\Sigma}|}\big(\frac{1}{2}\varepsilon_{i}^{\; jk}\partial_{j}\Sigma_{k} +\text{Im}(\Psi^{\dagger}\partial_{i}\Psi) \big),  \label{Lorentzian loss and yau}
\end{equation}
 where $\Sigma_{i}=2i\Psi^{\dagger}t_{i}\Psi$. However, its use is problematic because Lorentzian spinors may be null even when they are not vanishing.

Our construction of 
 magnetic modes  proceeds differently. We  use Lemma \ref{zero-modes relation} to obtain  magnetic Dirac modes on $\Il\subset \R^{1,2}$ directly  from the vortex magnetic   modes \eqref{vortex config modes}  on $SU(1,1)$. 
 The  magnetic field in Minkwoski space  is  obtained from the magnetic field $F_{A'}$ on $SU(1,1)$ via pull-back with  the  inverse stereographic projection  $H$. The magnetic field lines are therefore the images, under stereographic projection, of the  fields lines on $SU(1,1)$.  While the field lines on $SU(1,1)$ are  all closed,  they also  leave the domain of the stereographic projection. As a result, the image curves in   $\Il$ are not closed.  Instead, they are of  the form shown in Fig.~\ref{integral curves fig}. 

For explicit formulae on Minkowski space, it  is convenient to work in vector notation where a one-form is expanded as $A=\vec{A}\cdot d\vec{x}$ on $\Il$,
and where magnetic two-forms are expressed in terms of vector fields according to \eqref{2formvfield}. In particular, 
the inhomogeneous term in  the equation \eqref{magneticmodedef}  governing  vortex magnetic  modes   pulls back to 
the two-form 
\begin{equation}
-\frac{1}{4}H^{*}(\lf{1}\wedge\lf{2})=-\frac{4\ell^{2}}{(\ell^{2}+r^{2})^{2}}\star_{\R^{1,2}}\vartheta^{0}=\frac{1}{2}\varepsilon_{ijk}b^{i}dx^{j}\wedge dx^{k}, 
\end{equation}
and the corresponding magnetic field is 
\begin{equation}
 \vec{b}=\frac{-4\ell^{2}}{(\ell^{2}+r^{2})^{3}}\begin{pmatrix}
\ell^{2}-r^{2}+2(x^{0})^{2}\\
2(x^{2}\ell-x^{1}x^{0})\\
-2(x^{1}\ell +x^{2}x^{0})
\end{pmatrix}. \label{background_field}
\end{equation}
The field lines of $\vec{b}$ are  the fibres of the fibration $\lh:SU(1,1)\to H^{2}$,  and  plotted  in Fig.~\ref{integral curves fig}. 

\begin{figure}
\begin{center}
\vspace{-1cm}
\includegraphics[width=8truecm]{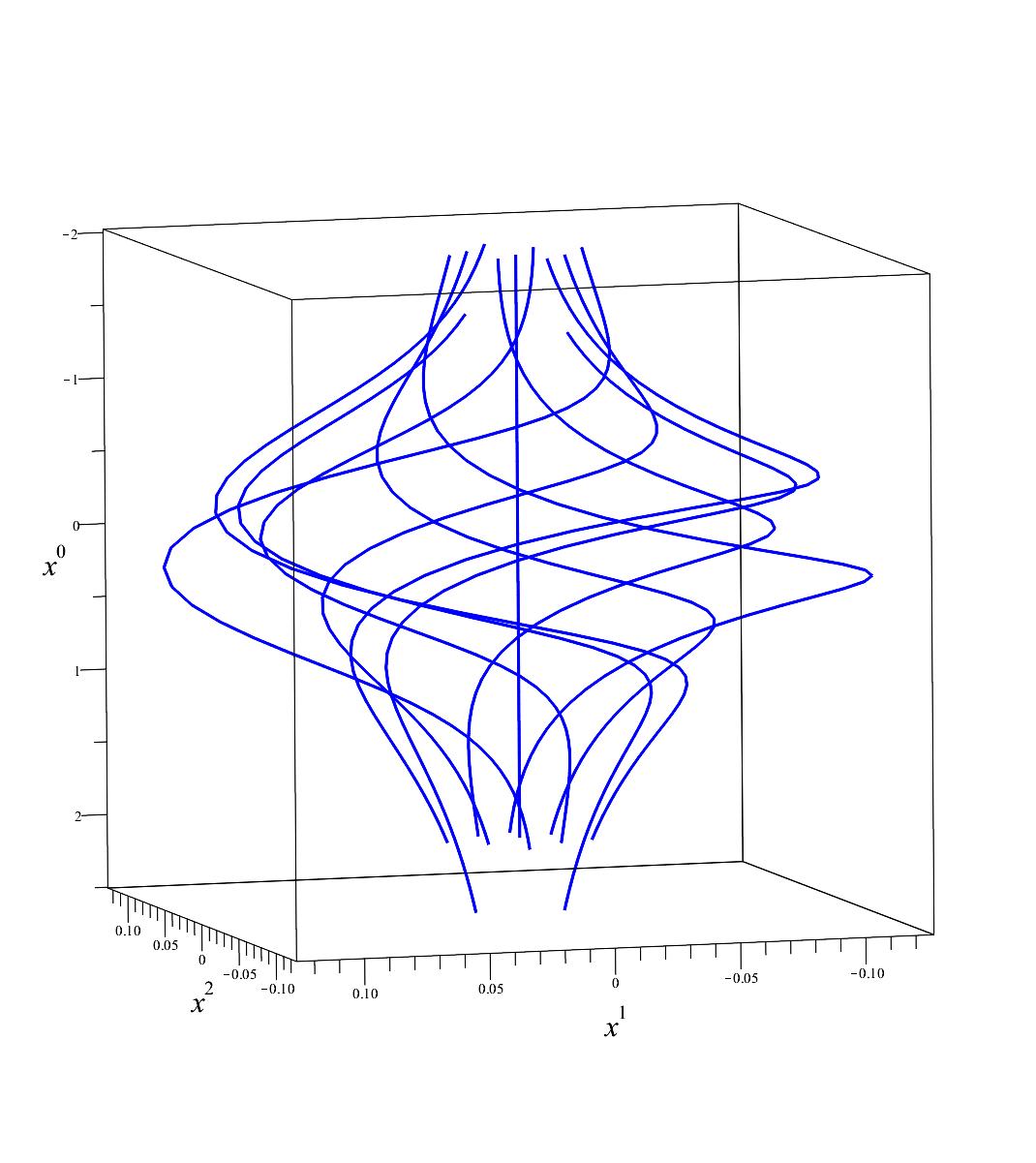}
\vspace{-1cm}
\end{center}
 \caption{The magnetic field lines for the pull-back of vortex magnetic modes to Minkowski space.  In particular, they  are the   magnetic fields lines of the background field $\vec{b}$. They are also the   images of the fibres illustrated  in Fig.~\ref{geodesics figure} under the stereographic projection.}
 
 \label{integral curves fig}
\end{figure}

Since vortex magnetic modes on $SU(1,1)$ satisfy a non-linear equation in addition to the linear Dirac equation, we expect the same to be true for the vortex magnetic modes on Minkowski space.  We  define them as follows.
\begin{dfn}
A pair $(\Psi, A)$ of a spinor $\Psi$ and a one-form $A=\vec{A}\cdot d \vec{x}$ on $\R^{1,2}$   is called a vortex  magnetic mode  in Minkowski space if it  satisfies the coupled equations
\begin{equation}
\slashed{D}_{\R^{1,2}, A}\Psi=0, \quad \vec{B}=-2i\Psi^{\dagger}\vec{t}\Psi+\vec{b}, \label{lorentzian cs eqs}
\end{equation}
where $\vec{B} = \nabla \times \vec{A}$ and $\vec{b}$ is the background field given in \eqref{background_field}
\end{dfn}
The coupled equations in this definition  formally  resemble the  dimensionally reduced Seiberg-Witten equations, perturbed by the background field $\vec{b}$. The role of the Seiberg-Witten equations in differential topology  makes essential  use of a Riemannian metric, and a Lorentzian version like the one defined here does not appear to have been studied.  

Combining many of the results derived in this paper, we arrive at the following explicit construction of vortex magnetic  modes on $\Il \subset  \R^{1,2}$: 
\begin{cor} \label{hol_functions_to_zero-modes}
Any given bundle map $V: SU(1,1) \rightarrow SU(1,1)$ covering a holomorphic map $f:H^2\rightarrow H^2$ determines a smooth vortex magnetic mode   on $\Il\subset \R^{1,2}$.  Explicitly, extracting the vortex configuration $(\Phi,A)$ on $SU(1,1)$ from $\mathcal{A}=V^{-1}dV$ via \eqref{A connection}, the vortex magnetic mode is given by
\bee
\Psi= G\begin{pmatrix}  \Omega^{-1} H^*\Phi \\ 0 \end{pmatrix} ,\quad A'_H= -H^*(A+\frac 34 \sigma^0).
\eee
\end{cor}
\begin{proof}
The  result follows by composing the construction of magnetic Dirac modes from vortex configurations with the construction of  vortex configurations from bundle maps.  We use  Theorem \eqref{vortex connection theorem} to construct a vortex configuration $(\Phi,A)$ on $SU(1,1)$ from the bundle map $V$, then Theorem \eqref{vortex config modes} to construct a vortex magnetic mode $(\Psi,A')$ on $SU(1,1)$  from $(\Phi,A)$. Finally Lemma \eqref{zero-modes relation} is used to pull it  back to $\Il$.  The confirmation that the magnetic mode thus obtained satisfies the coupled equations \eqref{lorentzian cs eqs} with  gauge field and magnetic field 
\bee
 A'_H= \vec{A}'_H\cdot d\vec{x}, \qquad \vec{B}'_H= \nabla \times \vec{A}'_H,
\eee
is a straightforward calculation, which is analogous to the one carried out 
for the   Euclidean  version  in \cite{RS}.
\end{proof}

The Corollary allows one to construct solutions of gauge Dirac equation and to solve initial value problems in Minkowski space. The restriction to $\Il\subset \R^{1,2}$ is not necessarily a problem in practice since $\ell$ can be chosen arbitrarily. By choosing it large enough, one can capture initial data on any bounded subset of a Cauchy surface. 

\section{Summary and outlook}
In this paper we presented a  lift of hyperbolic vortices to $\AAdS$ and  a  construction of  massless solutions of  magnetic Dirac equations on $\AAdS$ and  on a subset of  $\R^{1,2}$  from the vortices.  This  provides a new, three-dimensional interpretation of vortices and complements  the two-dimensional geometrical interpretation given by Baptista in \cite{Baptista} and the  four-dimensional interpretation as rotationally symmetric instantons \cite{Witten1,CD}. 

The summary diagram in  Fig.~\ref{summary}  gives a concise presentation of the spaces and equations that we considered here and the maps that relate them. The three dimensional   point of view  unifies  the massless Dirac modes and the hyperbolic vortices and clarifies  the geometry underlying this relationship.   
\begin{figure}[!htbp]
 \centering
\includegraphics[width=14truecm]{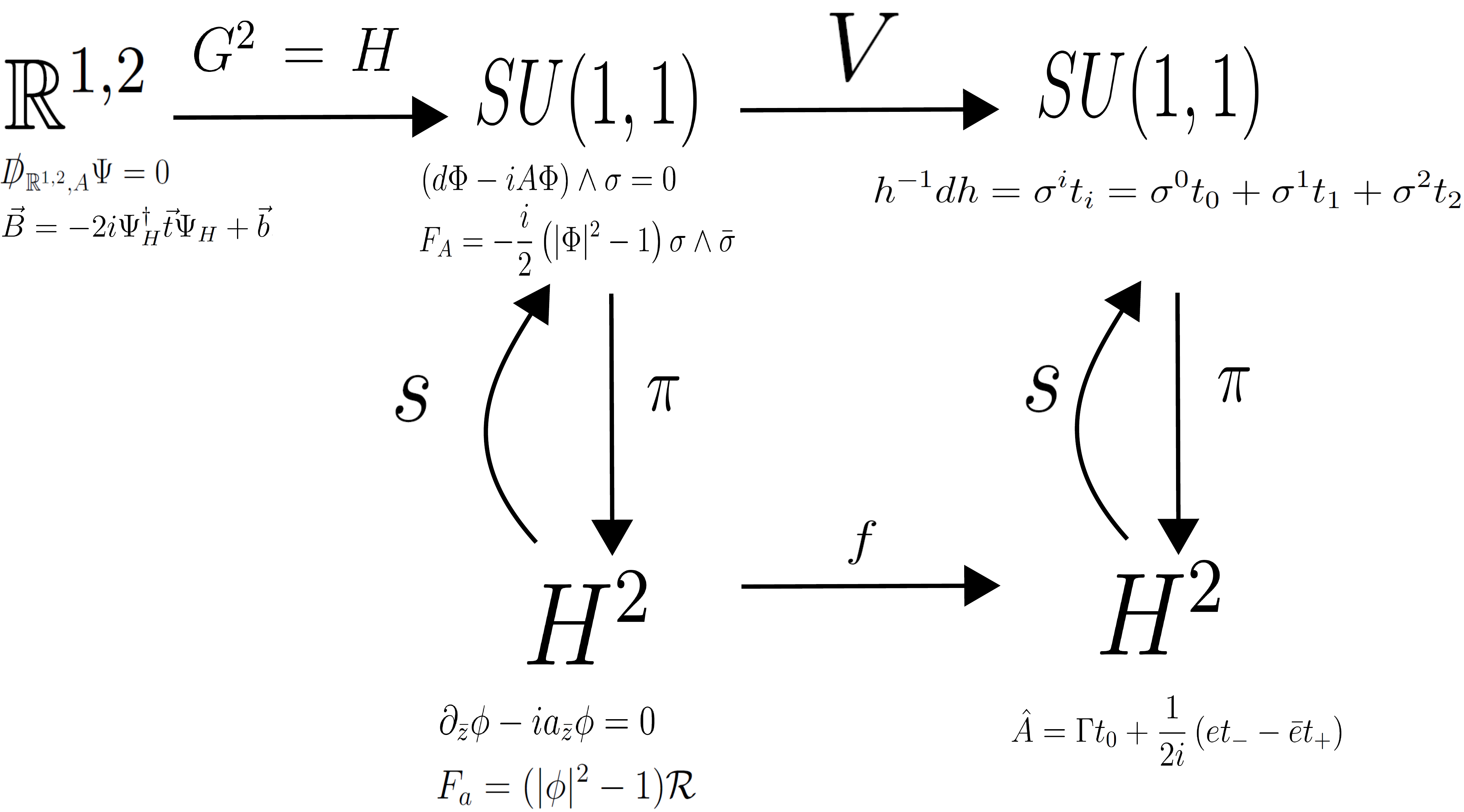}
 \caption{This Figure summarises all the spaces and equations considered here as well as the maps which relate them.}
 \label{summary}
\end{figure}

The story summarised in Fig.~\ref{summary} is  a Lorentzian and hyperbolic analogue of the Euclidean and spherical story told in \cite{RS}, but there important differences in the details. The triviality of $SU(1,1)$ as a circle bundle of $H^2$, as opposed to the non-triviality of the  Hopf bundle,    simplifies the topology   and allows for  global description of all sections. On  the other hand, the non-compactness of the base $H^{2}$, as opposed to the compactness of $S^2$,  leads to a wider variety of vortex configurations than in the Euclidean case, where vortices are in one-to-one correspondence with rational maps and the vortex number of any given configuration is finite.  

In the hyperbolic case, we have vortices  with a finite  vortex number, solved in terms of finite Blaschke products,  and configurations with  infinitely many zeros of the Higgs field. When  the latter are invariant under  a Fuchsian group  $\Gamma< SU(1,1)$ they lead to finite charge solutions of the vortex equations on the Riemann surface $H^2 / \Gamma$. The action of $\Gamma$ on $H^2$  descends from a left-action of $\Gamma$ on $SU(1,1)$  and therefore it makes sense to study vortex configurations  and spinors on $SU(1,1)$ which are left-invariant under $\Gamma$. However, the left-action of $\Gamma$ on $SU(1,1)$ does not, in general,  respect the domain of the  stereographic projection \eqref{stereo domain of definition}, so does not induce an action of $\Gamma$ on Minkowski space\footnote{The stereographic map intertwines the natural action of $\Gamma$ on Minkowski space with the conjugation action on $SU(1,1)$}. As result, there does not appear to be a natural characterisation of the  Dirac fields on  Minkowski space  obtained from hyperbolic vortices on  the Riemann surface $H^2 / \Gamma$. 

We end  with a brief outlook on interesting questions for further study. 
It seems very likely that all of the integrable vortex equations considered in \cite{Manton2} are amenable to a three dimensional interpretation along the lines of this paper and of \cite{RS}. However, the  cases studied here and in \cite{RS} show that there are  interesting differences in the details and in the interpretation, and these should be worked out. 

As observed  in \cite{CD}, all the integrable vortex equations considered in \cite{Manton2} can be seen as  dimensional reductions of the self-duality equations, with  gauge groups depending on the type of vortex. For example,  Popov vortices arise from $su(1,1)$-valued connections in four dimensions and hyperbolic vortices from $su(2)$-valued connections. The Cartan connections  we 
used are also non-abelian, but  exchange the Lie algebras, so use $su(2)$ for Popov vortices and $su(1,1)$ for hyperbolic vortices.  
It would be interesting to understand the link between the self-dual and the Cartan view point systematically for all the integrable vortices in \cite{Manton2}.

\noindent {\bf Acknowledgements} \, CR acknowledges an EPSRC-funded PhD studentship.

\end{document}